\documentclass[letterpaper]{article} 
\usepackage{aaai24}  
\usepackage{times}  
\usepackage{helvet}  
\usepackage{courier}  
\usepackage[hyphens]{url}  
\usepackage{graphicx} 
\usepackage{natbib}  
\usepackage{caption} 
\usepackage{algorithm}
\usepackage{algorithmic}

\usepackage{newfloat}
\usepackage{listings}
\DeclareCaptionStyle{ruled}{labelfont=normalfont,labelsep=colon,strut=off} 
\floatstyle{ruled}
\newfloat{listing}{tb}{lst}{}
\floatname{listing}{Listing}

\usepackage{subcaption}  
\usepackage{hhline}
\usepackage{multirow}
\usepackage{placeins}
\usepackage{amssymb}

\usepackage{amsthm}

\newtheorem{theorem}{Theorem}

\newtheorem{definition}{Definition} 

\usepackage[dvipsnames]{xcolor}

\title{From Space-Time to Space-Order: Directly Planning a Temporal Planning Graph by Redefining CBS}
\author{
    Yu Wu, Rishi Veerapaneni, Jiaoyang Li, Maxim Likhachev
}
\affiliations{
    Carnegie Mellon University \\
    \{yuwu3, rveerapa\}@andrew.cmu.edu, \{jiaoyanl, maxim\}@cs.cmu.edu
}

\begin{document}

\nocopyright
\maketitle

\begin{abstract}
The majority of multi-agent path finding (MAPF) methods compute collision-free space-time paths which require agents to be at a specific location at a specific discretized timestep. However, executing these space-time paths directly on robotic systems is infeasible due to real-time execution differences (e.g. delays) which can lead to collisions. To combat this, current methods translate the space-time paths into a temporal plan graph (TPG) that only requires that agents observe the order in which they navigate through locations where their paths cross. However, planning space-time paths and then post-processing them into a TPG does not reduce the required agent-to-agent coordination, which is fixed once the space-time paths are computed. To that end, we propose a novel algorithm Space-Order CBS that can directly plan a TPG and explicitly minimize coordination. Our main theoretical insight is our novel perspective on viewing a TPG as a set of \textit{space-visitation order} paths where agents visit locations in relative orders (e.g. 1st vs 2nd) as opposed to specific timesteps. We redefine unique conflicts and constraints for adapting CBS for space-order planning. We experimentally validate how Space-Order CBS can return TPGs which significantly reduce coordination, thus subsequently reducing the amount of agent-agent communication and leading to more robustness to delays during execution.

\end{abstract}

\section{Introduction}\label{sec:intro}
Multi-agent systems were originally constrained to computer games and simulations. However, with the rise of cheap and scalable robotics, there are increasing numbers of real-world multi-agent systems, in particular in drone swarms, robotic warehouses, and search and rescue systems. 

Multi-Agent Path Finding (MAPF) tries to find collision-free trajectories that usually minimize the net travel time of all the agents. Typically, MAPF methods find ``space-time" paths which require agents to be at specific locations at specific timesteps. Additionally, these methods typically do not reason about velocity or kinematic constraints and require the agents to traverse exactly one unit every timestep. Thus, these space-time paths are impossible to directly follow on real systems and do not allow agents to move flexibly. 

The current state-of-the-art method to allow agents to flexibly follow these space-time paths is to convert these space-time paths to a temporal plan graph (TPG) \cite{honig2016tpg}. The main idea is that agents can travel at arbitrary speeds as long as they visit overlapping locations in the same ``temporal precedence ordering" (i.e. relative sequence) as their space-time paths. For example, if there are 3 agents $A,B,C$ which have space-time paths that cross some location $s$ at timesteps 32, 14, 35 respectively, we remove the time dependency and just require them to visit $s$ in the relative order $B,A,C$. Real-world agents can now travel more robust to delays or faster velocities, and only need to wait and coordinate at these overlapped locations.
The current pipeline to plan for a real-world multi-agent planning system is then to plan space-time paths and then post-process it into a TPG, which the real-world robots can flexibly follow.

\begin{figure}[t!]
    \includegraphics[width=0.48\textwidth]{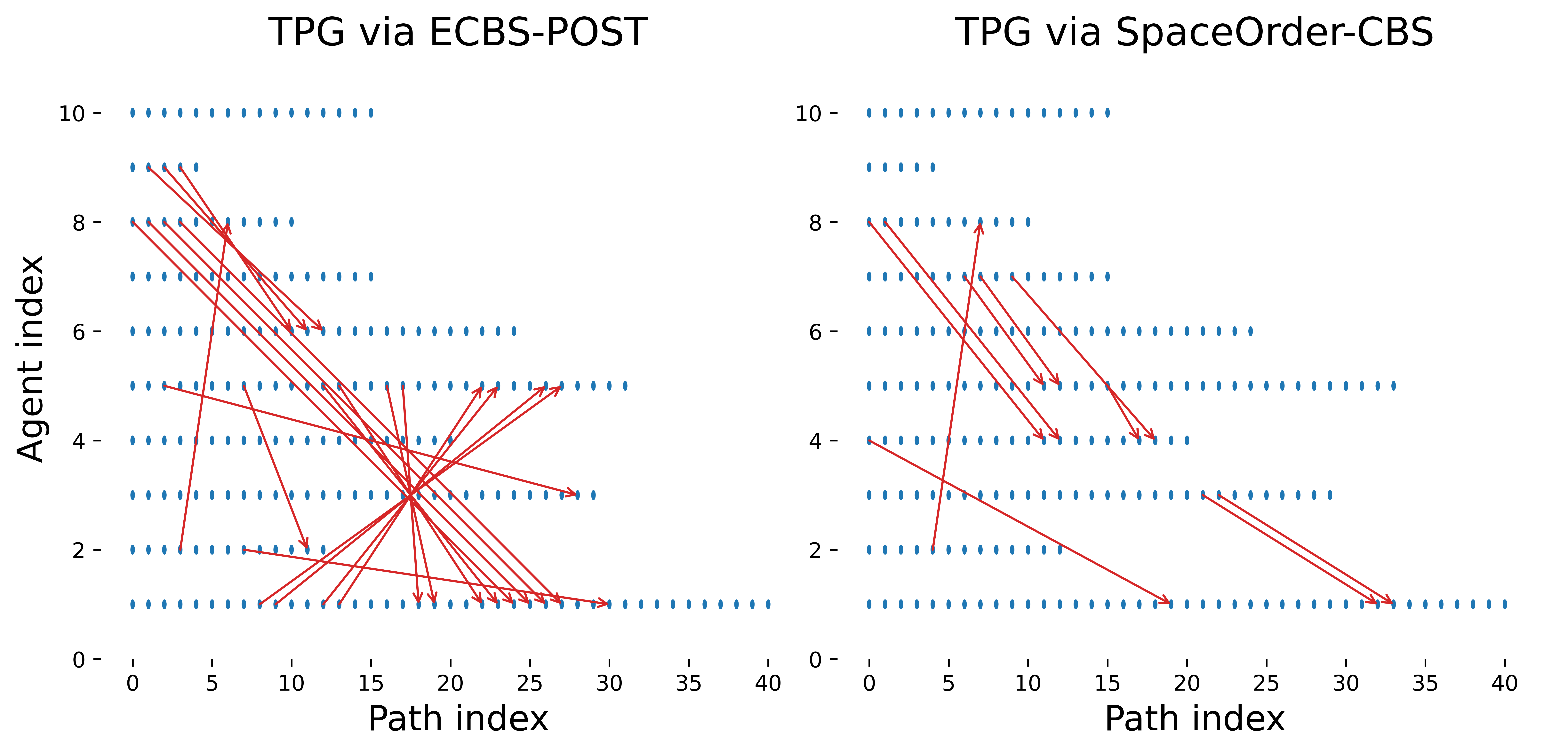}
    \vspace{-1.5em}
    \caption{TPGs on the random-32-32 map. Each row corresponds to an agent's path with each blue dot being a location on the path. Red arrows are Type-2 TPG edges between different agents, indicating coordination and potential waiting. Our method significantly decreases the number of Type-2 edges in the resultant TPG and is more robust to scenarios with delays. Note ``real" TPGs from actual MAPF plans have never been visualized before.}
    \label{fig:motivating}
    \vspace{-1.5em}
\end{figure}

However, this two-step process means that we cannot reason about the coordination required in the TPG that the agents follow, as once the space-time paths are computed, the TPG post-processing is fixed. Prior work \cite{wagnerFLOW} shows that the independently computed space-time paths can have high congestion and result in TPGs that require large coordination. Instead, if we could directly plan a TPG that agents could follow, we could directly reason about and minimize the amount of coordination that the agents require. Reducing coordination minimizes the amount of inter-agent communication needed which is a constraint in real-world multi-agent systems \cite{communicationReview2022} and can potentially improve robustness to delays.

Our core theoretical insight is that we can view an agent's path in a TPG as a set of spaces with relative visitation orders (e.g. visiting $s$ first vs visiting $s$ second). A TPG then consists of a set of these \textit{space-visitation order} paths which satisfy certain criteria. Armed with this knowledge, we can then adapt the general Conflict-Based Search (CBS) framework to create Space-Order CBS by defining unique constraints. SO-CBS can compute space-order paths that satisfy these criteria and produces a valid TPG directly. 

The main potential practical benefit of this method is that we now have the ability to plan paths that explicitly reason about coordinating with other agents. Specifically, when planning a single agent's path in the TPG, we know exactly how many agents we are waiting for at each location. Therefore, instead of possibly waiting for many agents at a location, we can find a different path that reduces the number of agents we need to wait for. This allows us to plan paths \textit{while simultaneously} minimizing coordination which is not possible when planning in space-time and post-processing afterward.
Succinctly, our main contributions are:





\begin{enumerate}
    \item Reinterpreting a TPG in a new perspective as a set of space-order paths that satisfy certain criteria.
    \item Designing Space-Order CBS, which adapts the Conflict-Based Search framework and defines unique constraints while maintaining completeness, to directly plan a TPG.
    \item Reasoning about per-agent coordination \textit{during} the planning process, and demonstrating how this can produce robust TPGs with significantly less coordination.
\end{enumerate}




\section{Related Work}\label{sec:related}
There are many MAPF works that plan space-time paths such as EECBS \cite{li2021eecbs}, BCP \cite{bcp2022}, and MAPF-LNS2 \cite{li2022mapf-lns2}. These methods all assume agents travel at unit velocities and during execution exactly follow their space-time paths. It is initially unclear how to adapt these space-time paths to be used by real robots which can incur execution changes/delays in practice (e.g. a robot temporarily breaking down).


\begin{figure*}[t]
    \centering
    \includegraphics[width=1\textwidth]{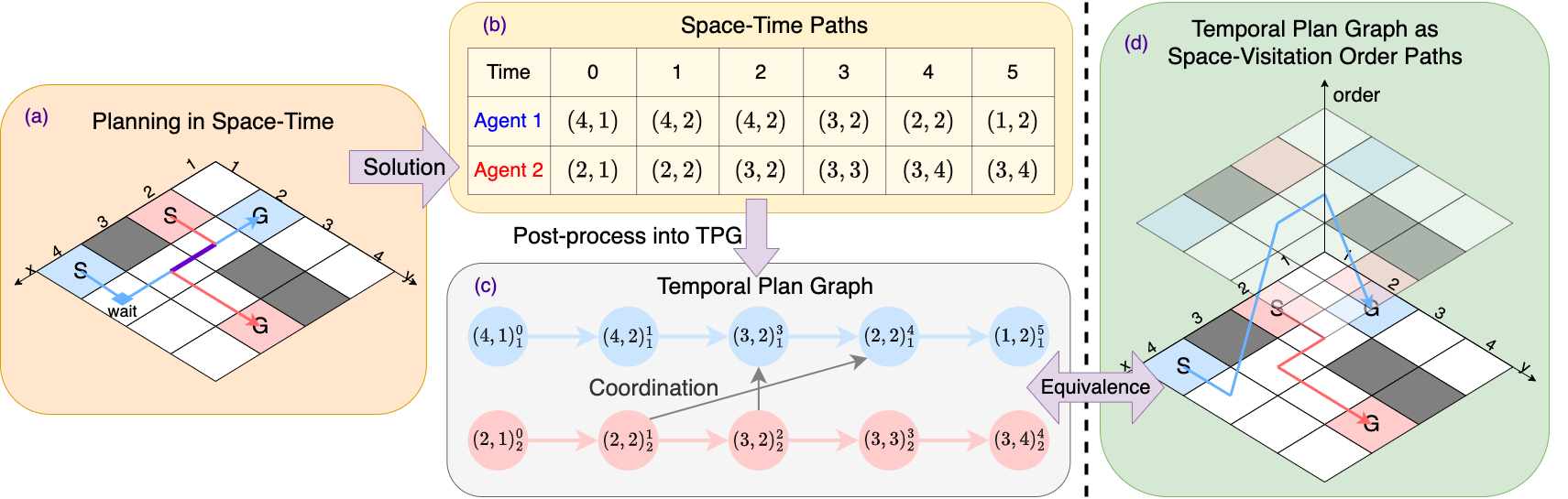}
    \caption{Given a MAPF problem with two agents (a), existing methods compute non-conflicting space-time paths that require agents to be at a specific location at a specific time (b). (a) shows how space-time paths may require agents to wait (e.g. the blue agent waiting at (4,2)), and have spatial overlaps in their paths (denoted by the purple region). To allow agents to handle delays or move at arbitrary velocities, the output space-time paths can be post-processed into a Temporal Plan Graph which removes waits and defines a precedence order (c). The TPG consists of vertices $s^t_j$ representing agent $j$ at location $s$ at timestep $t$, and edges denoting dependencies. Cross edges (gray) in the TPG from $s^t_j \rightarrow s^{t'}_k$, define possible coordination instances where agent $k$ needs to wait for agent $j$ to traverse location $s$ first. These occur at overlapping regions (e.g. purple in (a)) which require coordination to avoid collisions if we want agents to move independently and not timestep aligned. For example in (c), $(3,2)^2_2 \rightarrow (3,2)^3_1$ means that agent 1 needs to wait for agent 2 to leave $(3,2)$ before it can enter. In this paper, we show that by redefining vertices to space-visitation order $(s,r)$ vertices representing the $r^{th}$ agent to visit location $s$, a TPG in MAPF is equivalent to the agents having a set of non-interesting space-visitation order paths (d). This equivalence allows us to directly plan a TPG.}
    
    \label{fig:tpg-example}
    \vspace*{-1.5em}
\end{figure*}

MAPF-POST \cite{honig2016tpg} solves this problem by introducing the concept of a temporal plan graph (TPG). At its core, given a set of space-time paths, MAPF-POST constructs a TPG by removing the temporal component and just maintaining the relative precedence order (i.e. sequence) of agents visiting overlapping locations. During execution, agents are allowed to run arbitrarily fast or slow as long as they maintain their relative order as demonstrated in Figure \ref{fig:tpg-example}. Our key contribution is showing that instead of planning space-time paths and post-processing to a TPG, we can directly plan a TPG by viewing it from a unique perspective.


Explicitly reasoning about possible coordination in MAPF has generally been underexplored. MCP \cite{mapfdelayprob} post-processes a given TPG to reduce communication requirements. Concurrent work introduces a Bidirectional TPG \cite{su2024bidirectional} which also post-processes a given TPG to reduce communication requirements. These works are orthogonal of our objective of trying to directly compute a TPG with less coordination instead of post-processing them. A few works explicitly reason about K timestep delays by planning ``K-Robust" space-time paths \cite{atzmon2018krobust}, but are at the end of the day still planning space-time paths. \cite{offlinetimeind} introduces a new variant of MAPF problem where agents must find paths such that they can be traversed regardless of how agents are scheduled. This is over-restrictive in practice as it assumes agents cannot coordinate at all with each other, e.g. this problem setting will never allow two agents to cross a corridor at different timesteps and will fail if this is necessary.

The closest method which inspired our work is Space-Level Conflict-Based Search (SL-CBS) \cite{wagnerFLOW}. SL-CBS plans in ``space-levels" and from $(s,l)$ can choose to remain in the same level $(s',l)$ or increase one level to $(s,l+1)$. Within each level $l$, agents have collision-free path segments that are independent of each other and can move arbitrarily fast or slow without coordination.
Thus agents only require coordinating at transitions between different levels.
However, SL-CBS does not reason about the number of agents at each coordination instance and instead requires each agent to wait on all other agents. 
By planning a TPG, we can reason about which specific agents to wait on rather than waiting on all of them.

\section{Preliminaries}
\subsection{Space-Time MAPF Plan}
In a MAPF scenario with a total of $K$ agents, a space-time MAPF plan is a set of $K$ paths $\{ P_j | j=0,1,...,K-1\}$. Each path $P_j$ for agent $j$ is defined as an ordered sequence of space-time vertex $P_j = [ v_j^0, v_j^1, ..., v_j^\tau, ..., v_j^T]$, where $\tau$ is the index of vertex $v_j^\tau$ in the path sequence, $T$ is the last index in the path. Each vertex can be represented as a space-time tuple $(s, t)$, where $s$ denotes the space (i.e. location) and $t$ denotes the time step. In the space-time framework, agents $\tau=t$. We explicitly discriminate between path index and time for reasons that will be clear when we introduce space-order paths in Section \ref{sec:allcriteria}. 
The vertices in a path follow the transition rule that given a vertex $v_j^\tau = (s, t)$, the next vertex is $v_j^{\tau + 1} = (s', t+1)$, where $s'$ can either be a neighboring location of $s$ or $s$ itself. Such a space-time MAPF plan is a valid solution to a MAPF scenario if it is collision-free and all paths start and end at desired locations.


\subsection{Edge-Perspective TPG}
\citet{honig2016tpg} describes how to construct a TPG with ``Type-1" and ``Type-2" edges given a valid space-time MAPF plan defined above.
Conceptually, a TPG is a dependency graph that defines relative sequence of actions agents must follow. Each vertex $v^\tau_j$ defines the event of agent $j$ visiting its $\tau^{th}$-index location $s(j, \tau)$. 
Each edge denotes a dependency, e.g. the edge $v^2_j \rightarrow v^3_j$ represents how agent $j$ must first visit its $2^{nd}$ location before its $3^{rd}$ one. ``Type-1" edges are defined along an agent's path representing how it must travel its path sequentially. Thus, each agent has a corresponding set of vertices and Type-1 edges based on its path with waits removed, as depicted by the blue and red rows in Figure \ref{fig:tpg-example}(c).

The second relative sequence of actions occurs when two agents $j \neq k$ traverse the same location $s$ at different times $t_j < t_k$. They add a ``Type-2" edge from $(s, t_j) \rightarrow (s, t_k)$, denoting how the second agent $k$ must wait for the first agent $j$ to cross $s$ first. 

Note that Type-2 edges disallow following. If the first agent $k$ leaves location $s$ at timestep $t_{exc}$ during execution, the second agent $j$ can only enter location $s$ starting from $t_{exc}+1$.
Figure \ref{fig:tpg-example}(c) shows two examples of these Type-2 edges which define where agent-agent coordination is required to prevent collisions. Agents maintaining the TPG dependencies are guaranteed to reach their targets without deadlock or collision.

\section{Space-Order Perspective of a TPG}
This section formally shows how a TPG can be interpreted as a set of space-visitation order paths that satisfy specific criteria. Afterwards, we formally define coordination in the TPG context and relate it to space-visitation order. 

\subsection{Vertex-Perspective TPG}
\label{sec:allcriteria}
Instead of viewing a TPG from an ``edge" perspective as is done in \citet{honig2016tpg}, we choose to redefine its vertices and view it from a vertex perspective. Instead of a vertex $v^\tau_j = (s, t)$, we define vertices $v^\tau_j = (s,r)$, where $r$ represents the visitation order. For example, $v^\tau_j = (s, r=0)$ means agent $j$ is the first agent that visits $s$; $r = 1$ corresponds to the second agent to visit, and so on. An agent with a higher order must wait for all other agents who have a lower order at the same location. For example, if agent $j$ and $k$ visit $s$ at order $r_j$ and $r_k$ respectively, $r_j < r_k$ would imply agent $k$ must wait for agent $j$ to cross $s$ before it can enter $s$. Throughout the paper, we will use subscripts of $r$ to index the agent, and superscripts of $r$ to denote different locations. 

A vertex-perspective TPG is then a set of space-order paths $\{P_j\}$ that satisfy specific criteria. Each path is a sequence of space-order vertices $P_j = [ v_j^0, v_j^1, ..., v_j^\tau, ..., v_j^T]$, where $v_j^\tau = (s, r)$. The vertices in a path follow the transition that given a vertex $v_j^\tau = (s, r)$, the next vertex is $v_j^{\tau + 1} = (s', r')$, where $s'$ must be a neighboring location of $s$, $r'$ could be greater than, less than, or equal to $r$. 

A key difference between space-order paths compared to space-time paths is the transition rules. In the space-time case, the time step of vertices increases at a constant rate, implying that the order that agents visit each location is determined solely by the path length. In comparison, there is no restriction on the transition of order $r$ in space-order paths. This flexibility enables us to explore different possible visitation orders at each location during planning. In addition, unlike space-time paths that explicitly contain a waiting action, waiting action is not contained in the space-order transition rule, as it is modeled implicitly by the coordination between agents. Figure \ref{fig:tpg-example}(d) shows how the TPG in (c) can be reinterpreted as $(s,r)$ vertices.


A valid TPG is a TPG that ensures all agents can reach their target locations without any collision or deadlock. To formalize the conditions of a valid TPG, we define the following criteria. Firstly, the most obvious condition is that two distinct agents cannot share the same order at the same location. Otherwise, they might reach that location at the same time during execution and collide.
\begin{definition}[Vertex Criteria] \label{def:cVertex} Agents must have unique $(s,r)$ vertices in their paths.
\end{definition}

Since each agent starts at their start location, they must hold the lowest order of their start locations. For similar reasons, each agent $j$ should also hold the highest order at their goal location where they rest indefinitely. Otherwise, this means there is another agent $k$ that visits the goal location after $j$ rests at it, which is a contradiction. Thus a TPG's space-order paths must satisfy the following Start and Target Criteria:
\begin{definition}[Start Criteria] \label{def:cStart}
An agent $j$  must be the first agent to occupy it's start location $s_j^{start}$ compared to all other agents $k$. Thus $r_j < r_k, \forall k \neq j$.
\end{definition}
\begin{definition}[Target Criteria] \label{def:cTarget}
An agent $j$ resting at its goal location $s_j^{goal}$ must be the last agent to reach there compared to all other agents $k$. Thus $r_j > r_k, \forall k \neq j$.
\end{definition}

\begin{figure}[t!]
    \centering
    \includegraphics[width=0.47\textwidth]{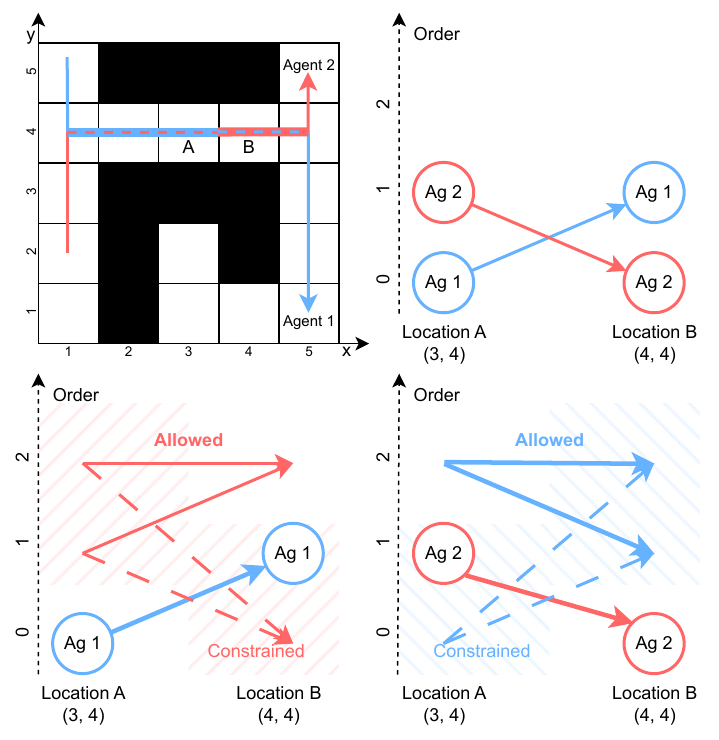}
    \caption{Example edge conflict. Agent 1 (blue) and agent 2 (red) need to cross a corridor in the same direction. The top-left subplot shows the paths of two agents. Solid line segments represent the corresponding agent hold order $r=0$ at the location; dashed line segments represent the corresponding agent hold order $r=1$. In the first half of the corridor, agent 2 has a higher order and must follow agent 1. However, the relative order of two agents swap at the edge $(3,4) \rightarrow (4,4)$. We label them as location $A$ and location $B$. The top-right subplot shows the space-order vertices at that particular edge. This edge has an edge conflict as during execution agent 1 wait at location $A$ for agent 2 to visit $B$, but agent 2 would be stuck behind agent 1 before $A$; thus both agents would wait forever. The bottom row shows two edge constraints to resolve this conflict. The dashed edges are forbidden while the solid edges are allowed. If both ends of an arrow are within the shaded area, the corresponding edge is disallowed.}
    \label{fig:edge_conf}
    \vspace{-1em}
\end{figure}


Now we consider possible deadlock situations. Consider the example depicted in the top-left subplot of Figure \ref{fig:edge_conf}, where agent 1 and agent 2 traverse the same edge $A \rightarrow B$ in the same direction. Agent 2 (red) starts with a higher order than Agent 1, and is thus following it, up till location $A$. However at $B$, all of a sudden Agent 1 starts following 2. This is physically impossible as this means that Agent 2 overtook Agent 1 between $A$ and $B$, but there exists no intermediate location (since they are consecutive locations) for this overtake to take place. Due to this occurring from transitions between two consecutive locations ($A, B$), we define the following Edge criteria which a valid TPG must have to avoid such scenarios.


\begin{definition}[Edge Criteria] \label{def:cEdge} 
Case 1: If both agents $j, k$ traverse $s, s'$ consecutively, i.e. $j$ has $(s,r_j) \rightarrow (s',r'_j)$ and $k$ has $(s,r_k) \rightarrow (s',r'_k)$, their relative order at $s$ and $s'$ must be the same, i.e. $sgn(r_j - r_k)$ must equal $sgn(r'_j - r'_k)$.
Case 2: Similarly, if two agents swapping consecutive locations, i.e. $j$ has $(s,r_j) \rightarrow (s',r'_j)$ and $k$ has $(s',r'_k) \rightarrow (s,r_k)$ (note $s,s'$ swapped), then $sgn(r_j - r_k)$ must equal $sgn(r'_j - r'_k)$.
\end{definition}

\begin{figure}[t!]
    \centering
    \includegraphics[width=0.47\textwidth]{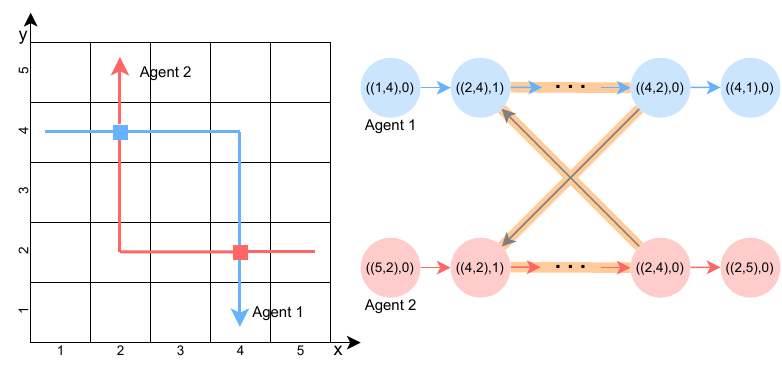}
    \vspace{-1.2em}
    \caption{Example deadlock cycle. The left subplot shows the space-order paths of two agents. All vertices in the paths have default order $0$, except that vertices marked with a colored square represent the corresponding agent has order $1$ at that location. e.g. at $(2,4)$ agent 1 has order $1$ and agent 2 has order $0$. The right subplot shows the corresponding TPG. Colored edges are Type-1 edges. Grey cross-edges are Type-2 edges specifying dependency at overlapping locations. There exists a cycle $((2,4),1) \rightarrow ... \rightarrow ((4,2),0) \rightarrow ((4,2),1) \rightarrow ... \rightarrow ((2,4),0) \rightarrow ((2,4),1)$ which we highlight in orange. The cycle consists of two Type-2 edges and two sequences of Type-1 edges. The cycle implies a deadlock. Agent 1 will stuck waiting to enter $(2,4)$, while agent 2 will stuck waiting to enter $(4,2)$.  }
    \label{fig:deadlock_example}
    \vspace{-1em}
\end{figure}

Deadlocks could happen at non-neighboring locations as well. Existing work \cite{berndt2020feedback} proves cycles in TPG lead to deadlocks. 
Consider the space-order paths of two agents. Agent 1 needs to wait for agent 2 at location $(2,4)$, while agent 2 must wait for agent 1 at $(4,2)$. However, agent 1 cannot reach $(4,2)$ to resolve the waiting for agent 2. 
The same applies to agent 2, leading to both agents being stuck in a deadlock. The fundamental reason is that there exists a cycle in the TPG. Since an edge in TPG represents dependency between vertices it connects, all vertices inside the cycle represent an event that is a pre-requisite of itself. Those events can never happen by the definition of TPG, implying a deadlock. Such deadlocks can involve more than two agents. We formally define it as:
\begin{definition}[Cycle Criteria] \label{def:cDeadlock}
A deadlock cycle of $N$ agents consists of $N$ Type-2 edges (i.e. coordination) at $N$ different locations.
Let $S = \{s^n | n = 0, 1, ... N-1 \}$ be a set of $n$ distinct locations and $A = \{a_n | n = 0, ... N-1\}$ a set of distinct agents. Index $n$ is in module $N$. Let $r_m^n$ be the order of agent $a_m$ at location $s^n$. Sets $S$ and $A$ form a deadlock cycle if
\begin{enumerate}
    \item $\forall n$, agents $a_n, a_{n-1}$ visit $s^n$ and $r_n^n > r_{n-1}^n$. That is, agent $a_n$ must wait for agent $a_{n-1}$ to enter location $s^i$.
    \item $\forall n$, agent $a_n$ visits $s^n$ then $s^{n+1}$. Combined with the previous condition, this implies agent $a_n$ must first wait for agent $a_{n-1}$, then enter $s^n$, and finally visit $s^{n+1}$.
\end{enumerate}
\end{definition}
By condition 1, agent $a_{N-1}$ must wait for $a_{N-2}$ to leave $s^{N-1}$ before it can enter due to the Type-2 edge. By condition 2, $a_{N-2}$ must \emph{first} visit $s^{N-2}$ before going through $s^{N-1}$ so that $a_{N-1}$ moves on. However, by condition 1, $a_{N-2}$ must wait for $a_{N-3}$ at $s^{N-2}$. Similarly applying condition 2 and 1 iteratively, $a_{N-3}$ must wait for $a_{N-4}$ at $s^{N-3}$, etc. This chain of waiting goes on and eventually leads to $a_0$ waiting for $a_{N-1}$ at $s^0$, but remember $a_{N-1}$ is stuck waiting from the beginning! This forms a cycle of agents stuck waiting for each other at locations $S$.

If an agent has $(s,r=1)$, it is the 2nd agent to visit $s$ and there must be a different agent $k$ which visited $s$ first and occupies $(s,r=0)$. Intuitively, a space-order TPG should not contain any ``empty slots" .
\begin{definition}[Relative Criteria] \label{def:cRelative}
If an agent $j$ has an order $r_j > 0$ at some location $s$, there must be another agent $k$ with order $r_j-1$ at location $s$. 
\end{definition}

Altogether, from this vertex perspective, a valid space-order TPG consists of (location, visitation order) vertices that start and end at desired locations, and satisfy all criteria. Note regular space-time planning has a similar vertex and target criteria, but that the others are unique to TPGs.

\subsection{Coordination in TPG} \label{sec:coordination}
Coordination between agents is required when multiple agents visit a location. If agent $j$ visits vertex $(s, r)$, it need to wait for, or coordinate with $r$ agents at location $s$. One natural objective is to minimize the total number of these possible wait interactions. In the edge-perspective TPG proposed by \cite{honig2016tpg}, each wait interaction is represented by a Type-2 edge. Therefore, this objective is equivalent to optimize the total number of Type-2 edges in TPG. For example Figure \ref{fig:tpg-example}(c) contains two such instances, requiring that the red agent must traverse $(2,2)$ and $(3,2)$ first before the blue agent.

\begin{definition}[Total Coordination] \label{def:cTotal}
Given a single space-order vertex $v_j^\tau = (s, r)$, its total coordination is $C_{total}(v_j^\tau) = r$. The total coordination of a TPG is the sum of the coordination of all vertices in all agents' paths: 
$C_{total} = \sum_j \sum_\tau C_{total}(v_j^\tau)$.
\end{definition}

This objective could be over-penalizing certain coordination, though. For example, if agent $j$ and $k$ shares a long segment of spatial path of length $L$, $C_{total}$ will penalize this situation by $L$. This could be unnecessary. Therefore we propose unique coordination as an alternative objective, which will only penalize the above situation by a value of $1$.
\begin{definition}[Unique Coordination] \label{def:cUnique}
For agent $j$'s path $P_j$, unique coordination $C_{unique}(P_j)$ is the number of other agents whose paths overlap with $P_j$ at least once. The unique coordination of a TPG is $C_{unique} = \sum_j C_{unique}(P_j)$.
\end{definition}

In the example shown in Figure \ref{fig:tpg-example}, $C_{total}=2$ while $C_{unique}=1$. In the rest of the paper, we represent coordination as $C$ in general, while $C$ can be either $C_{total}$ or $C_{unique}$.

\section{Space-Order CBS}
We have shown that we can view the TPG as a set of space-order paths which satisfy certain criteria and we have defined coordination in a TPG. Instead of planning space-time paths and post-processing to get the $(s,r)$ TPG vertices, our objective is to directly plan a set of space-order paths that represent a valid TPG. During planning, we can minimize coordination by reasoning about $C_{total}$ or $C_{unique}$. To that end, we designed Space-Order Conflict-Based Search (SO-CBS).

Conflict-Based Search is a popular weakly complete and optimal space-time MAPF planner \cite{sharon2015cbs}. The CBS framework employs a high-level constraint tree (CT) that searches over conflicts in a best-first manner, and a low-level A* search that searches individual paths while respecting constraints. 
Our method, Space-Order CBS, redefines the CBS framework to plan a TPG by planning conflict-free space-order paths instead of space-time trajectories. Following CBS's framework, we plan paths for agents individually and resolve conflicts (i.e. violations of Definition \ref{def:cVertex}-\ref{def:cRelative}) by adding constraints. We design novel space-order constraints that maintain completeness.
This section discusses the necessary components for Space-Order CBS.

\subsection{Low-Level Planner} \label{sec:lowlevelpl}
Our low-level planner must search for a valid space-order path. All agents are initialized at their start location with $r=0$. Unlike space-time where $(s,t)$ edges go to $(\{s,s'\},t+1)$, space-order can go from $(s,r)$ to $(s',r')$ where $r'$ can be $>,<,=$ to $r$ and $s' \neq s$.  We determine possible $r'$ by computing $R'_{max}$, the maximum order at $s'$ among all existing paths of other agents. The planner inserts all unconstrained $(s',r'), r' \in \{0, 1, ..., R'_{max}+1\}$. Note if $s'$ is the start location of another agent $j$, we omit $(s',r=0)$ since $j$ must start at $(s',0)$. Since each vertex can transit to multiple orders for the same neighbor, space-order search has a larger action space than space-time search.

The low-level planner finds the minimum-cost path for agent $k$ from start to goal location that satisfies the applied CT constraints defined in the next section. Section \ref{sec:coordination} defines the coordination cost $C_k$ which the low-level planner could minimize. However, this can significantly increase the path length and thus overall plan execution time. Therefore, we minimize $w * C(P_k) + (1-w) * |P_k|$, a weighted sum of coordination $C(P_k)$ and path length $|P_k|$. The hyper-parameter $w$ trades off the two with $w = 1$ representing purely minimizing coordination while $w=0$ represents solely minimizing path length. 

\subsection{Criteria, Conflicts, \& Constraints}
\label{sec:all-constraints}
Conflict-Based Search works by detecting conflicts that result in invalid paths and defining and applying constraints that resolve these conflicts. Section \ref{sec:allcriteria} defines the criteria that space-order paths need to satisfy to define a valid TPG; a conflict occurs when any of these criteria are violated. Each criterion thus defines a conflict and requires the following constraints to resolve. We do not have a constraint for the start criterion as we implicitly enforce it by initialization specified in Section \ref{sec:lowlevelpl}.

\begin{definition}[Vertex Constraint]
    This occurs when two agents $j,k$ share the same $(s,r)$ location, violating Definition \ref{def:cVertex}. We add a constraint on $j$ to avoid $(s,r)$ for one child node, and a constraint on $k$ likewise (similar to resolving a regular space-time vertex conflict).
\end{definition}

\begin{definition}[Target Constraint]
    This occurs when agent $k$ passes over agent $j$ resting at its goal location $(s_j^{goal},r_j)$ violating Definition \ref{def:cTarget}. We then constrain $j$ to reach and rest at $(s_j^{goal},r \leq r_j)$ and have no other agent (including $j$) with $(s_j^{goal},r > r_j)$, or constrain $j$ to reach and rest at $(s_j^{goal},r > r_j)$.
    \label{def:target-constraint}
\end{definition}

Now let's return to the example in the top-right subplot of Figure \ref{fig:edge_conf}. As we discussed in Definition \ref{def:cEdge}, the main problem of an edge conflict is that two agents have different relative priorities on two consecutive locations. To resolve it, our contraints should disallow swapping order. To that end, our constraint should encourage $r^A_1$ to increase and $r^B_1$ to decrease. So we disallow agent 1 to traverse edge $(s^1, r \le 2) \rightarrow (s^2, r \ge 2)$ (bottom-right of Figure \ref{fig:edge_conf}). Reasoning for agent 2 is similar. These constraints would still allow some edge conflicts to occur at the same locations, which would be resolved by a finite number of future constraints. This is necessary for the completeness of the search (see Theorem \ref{thm:completeness} and Appendix for details). In general, we define edge constraint as the following: 
\begin{definition}[Edge Constraint]
    In the case of Definition \ref{def:cEdge}.1, two agents $j,k$ swap relative orders in consecutive locations in the same direction, e.g. $(s,r_j) \rightarrow (s',r'_j), (s,r_k > r_j) \rightarrow (s',r'_k < r'_j)$
    , violating Definition \ref{def:cEdge}. 
    Instead, we add the following child constraints: 
    (a) Constrain $j$ to avoid the set of edges $(s, r \le r_j) \rightarrow (s',r \ge r'_k)$. 
    (b) Constrain $k$ to avoid set of edges $(s,r \ge r_j) \rightarrow (s', r \le r'_k)$.
    The analogous case in Definition \ref{def:cEdge} where two agents travel in opposite directions follows without loss of generality.
\end{definition}

Intuitively, a deadlock cycle can be resolved by breaking any of the edges in the cycle. In the example in Figure \ref{fig:deadlock_example}, we can resolve the cycle by reversing the visitation order at any of the coordination places. i.e. if agent 1 must wait for agent 2 at both $(4,2)$ and $(2,4)$, or vice versa, there will not be a deadlock. This is equivalent to reversing the direction of one of the Type-2 edges Alternatively, we can break the sequence of Type-1 edges $((2,4),1) \rightarrow ... \rightarrow ((4,2),0)$, i.e. disallow agent 1 to first visit $(2,4)$ then visit $(4,2)$. This also resolves deadlock. The same applies to agent 2. More generally, we define constraints to resolve deadlock cycle with $N$ agents as:
\begin{definition}[Cycle Constraints]
    Given a deadlock cycle in \ref{def:cDeadlock} involving $N$ Type-2 edges: $(s^n, r^n_{n-1}) \rightarrow (s^n, r^n_n)$ where $n = 0, 1, 2, ..., N-1$ and $r^n_{n-1} < r^n_n$, we can break the cycle if any of the edges is broken. So we add 2 constraints for each $i$:
    (a) Agent $n-1$ cannot go to $(s^n, r < r^n_n)$ if it enter location $s^{n-1}$ before $s^n$.
    (b) Agent $n$ cannot go to $(s^n, r \ge r^n_n)$.
    \label{def:deadlock-constraint}
\end{definition}
There are $2N$ possible constraints, resulting in $2N$ CT nodes. Comparing regular conflicts with a branching factor of $2$ in CBS, deadlock conflicts of $N$ agents lead to a $2N$ branching factor, which is very expensive to resolve.

\begin{definition}[Relative Constraint]
    Following Definition \ref{def:cRelative}, a relative conflict occurs when agent $j$ occupies $(s,r_j)$, but there is no agent at $(s,r-1)$ (e.g. dashed circles in Figure \ref{fig:edge_conf}). We must either constrain $j$ to avoid all $(s,r \ge r_j)$, or constrain another agent $k$ to use $(s,r_j-1)$ in its path.
\end{definition}
Maintaining the relative criteria is hard as the relative constraint needs to enforce some unknown agent $k$ must occupy $(s,r_j-1)$. Resolving this failed criteria for MAPF with $K$ agents requires generating $K$ constraints and children for the CT node, which would be computationally infeasible. Therefore, we ignore this criterion and allow an agent at $(s,r)$ while no other agent is at $(s,r-1)$. This TPG will still maintain a valid precedence order but the possible unused $(s,r)$ makes the coordination cost $C_{total}$ no longer accurate. Thus even though an optimal Space-Order CBS exists, a practical one cannot have these $N$ children.

\begin{theorem}[Completeness] \label{thm:completeness}
    Space-Order CBS is weakly complete (will eventually find a solution if it exists).
\end{theorem}
\begin{proof}
    Space-Order CBS's main proof of completeness comes by proving that if a solution exists in a CT node, for each conflict resolved via constraints, a solution exists in at least one child node. The proof for Vertex and Target Constraints follows regular space-time CBS. The proofs for Edge and Deadlock Constraints are shown to retain a solution in at least one of their child nodes in the appendix.
    Thus Space-Order CBS will always contain the solution in the CT as the root allows the solution (no constraints) and applying constraints retains a solution that exists in at least one child node. Since Space-Order CBS high-level and low-level searches are best first with respect to an increasing cost function, and a valid TPG has a finite cost, the search will eventually examine that solution.
\end{proof}

We re-emphasize that SO-CBS is also not optimal due to the relative constraints (Section \ref{sec:all-constraints}).
Nevertheless, we still use the CBS framework as it allows for fine-grained conflict resolution and is still complete and therefore searches through more paths. Thus SO-CBS produces better solutions than greedier/incomplete alternatives like prioritized planning which produce longer paths. 

\subsection{High-Level Planner}

The high-level planner chooses the minimum cost set of paths, detects conflicts (Section \ref{sec:allcriteria}), applies constraints (Section \ref{sec:all-constraints}), and generates child nodes by calling the low-level planner (previous paragraph \& Section \ref{sec:lowlevelpl}). Similar to the low-level planner, the high-level planner could minimize just coordination but runs into the same issues of long path lengths. Therefore it prioritizes the same weighted objective $w * C + (1-w) * \sum_k |P_k|$. 

\subsection{Modifying Space-Time CBS Improvements}
We adapt the following CBS improvements out-of-the-box for Space-Order CBS to speed up the search.

\emph{Enhanced CBS} speeds up search by finding $w_{so} > 1$ suboptimal solutions by using focal searches \cite{focal1982} in both the high-level and low-level planners to select nodes with fewer conflicts \cite{barer2014suboptimal}. We use this technique without modifications in Space-Order CBS.

\emph{Bypassing conflicts} is a technique that tries to reduce the branching factor of the CT \cite{bpEli2015}. When expanding a CT node, the bypass technique replaces the node with a child node instead of generating two branches iff the following conditions are true: (1) cost does not increase (2)  number of conflicts decreases. In suboptimal search, the first condition is replaced by cost of the child node being within the suboptimal bound. This technique effectively reduces the size of CT and thus speeds up the search. This logic can be directly applied to SO-CBS without modifications.

\emph{Target reasoning} speeds up resolving conflicts at goal locations by avoiding multiple expansions on conflicts at the same goal location \cite{srli2021}. We have directly incorporated it in Definitions \ref{def:cTarget} and \ref{def:target-constraint}.

\section{Improving SO-CBS}
It is not clear \textit{a priori} that Space-Order CBS will be particularly slow, but when running it we found that it struggles for small numbers of agents. Additionally, the produced TPGs are inefficient and contain many unnecessary waiting. Therefore, we propose two improvements to the algorithm: avoiding inversions and introducing negative order.

\subsection{Avoiding Inversions}
Simply optimizing the coordination objectives and path lengths does not guarantee to produce efficient TPG. 
Suppose agent $j, k$ visit location $s$ at space-order vertex $v_j^{\tau_j} = (s, r_j)$ and $v_k^{\tau_k} = (s, r_k)$, respectively. Assume $\tau_j < \tau_k$ but $r_j > r_k$. $r_j > r_k$ implies agent $j$ must wait for agent $k$. However, $\tau_j < \tau_k$ means agent $j$ travels less distance to enter $s$ than agent $k$. Agent $j$ is likely to arrive $s$ before agent $k$ and need to wait for $k$ for a long time! We define such situations as \textit{inversions}. They are inefficient and lead to longer execution times of our TPG.

\begin{definition}[Inversion] \label{def:inversion}
    Suppose agent $j, k$ visit location $s$ at space-order vertex $v_j^{\tau_j} = (s, r_j)$ and $v_k^{\tau_k} = (s, r_k)$. They constitute an inversion if boolean value $\tau_j < \tau_k$ does not equal to $r_j < r_k$.
\end{definition}

Additionally, these inversions happen to be necessary components of a deadlock conflicts. A deadlock conflict of $n$ agents contains $n$ inversions and requires creating $2n$ high-level nodes which is very slow. Albeit under a different problem, \cite{offlinetimeind} also observed that cycle detection and resolution is a significant bottleneck during the search in their method. Therefore, avoiding inversions can significantly reduce the execution time of TPG as well as the planning time. 

\subsubsection{Changing Focal Queue Penalty} \label{sec:focalq}
In order to avoid inversions, for each low-level node, we modified its cost in Space-Order CBS's low-level focal queue to $f_{focal} = n_c + P * n_{inv}$, where $n_c$ is the number of conflicts on the path up to that node, $n_{inv}$ is the number of inversions, $P$ is a constant penalty. Empirically, we found that setting $P$ to a large value of $1000$ is most effective in practice.

\subsection{Negative Orders}
The behavior of SO-CBS could depends on the order that it plans paths for agents. Since the low-level search tries to minimize the order, an agent $j$ will visit order $0$ by default if no other agents visit the same location $s$. However later when another agent $k$ needs to visit $s$, it will have to visit at a higher order $r>0$ or conflict with agent $j$ at $(s,0)$. There is no way for agent $k$ to have a lower order than $j$, due to the planning order. This problem becomes significant when we incorporate the inversion penalty as the low-level search for $k$ would penalize $r>0$ if it leads to an inversion. Thus instead the search would prioritize finding a longer path to avoid $s$ or instead conflict with $j$ at order $0$, resulting in longer paths or significantly more conflicts and slowing down the search.

To solve this issue, we allow agents to visit ``negative orders", so that agents can obtain the first priority without incurring a conflict even if order $0$ is occupied. This implies we need to re-interpret order $r$ in a relative rather than an absolute sense. Practically, this requires modifications to both the low-level transition rules and the computation of $C_{total}$.

\subsubsection{Relative Order}
With negative values for $r$, the definition of $r=0$ as the $1$-st agent to visit does not make sense anymore. Rather, we will interpret $r$ in a relative manner. At a certain location, the agents' visitation priority will be sorted by their $r$ value, with lower $r$ value representing higher priority. The relative relationship is unchanged: an agent with a higher order must wait for all other agents who have a lower order at the same location.

\subsubsection{Low-Level Transitions with Negative Orders} \label{sec:transitions}
In Section \ref{sec:lowlevelpl}, in a vertex trainsition from $(s,r)$ to $(s',r')$, the low-level planner inserts all unconstrained $(s',r'), r' \in \{0, 1, ..., R'_{max}+1\}$. We only need to change that to inserting all $r' \in \{R'_{min}-1, ..., 0, 1, ..., R'_{max}+1\}$, where $R'_{max}$ and $R'_{min}$ are the maximum and minimum order at $s'$ among all existing paths of other agents. Also, we initialize agents at their start locations with $r=R_{start}$, where $R_{start}$ is a large negative value and is the minimum possible order.

\subsubsection{Changing $C_{total}$ Calculation}
Now that the absolute value of $r$ does not have concrete meaning, we cannot compute $C_{total}$ as the sum of $r$ over all vertices as in Definition \ref{def:cTotal}. Luckily, $C_{total}$ is equivalent to the total number of Type-2 edges in the TPG. We can easily compute it by counting the Type-2 edges at each vertex and taking summation.

\begin{definition}[Total Coordination] \label{def:cTotal-relative}
Given a single space-order vertex $v_j^\tau = (s, r_j)$, its total coordination $C_{total}(v_j^\tau)$ is the number of other agents $k$ visiting $(s, r_k)$ with lower order $r_k < r_j$. The total coordination of a TPG is the sum of the coordination of all vertices in all agents' paths: 
$C_{total} = \sum_j \sum_\tau C_{total}(v_j^\tau)$.
\end{definition}

\subsubsection{Still Complete}
Since we are not changing any of the constraints, we can reuse the same proof of completeness and retain identical completeness guarantees. 

\begin{figure*}[t!]
    \centering
    \includegraphics[width=0.95\textwidth]{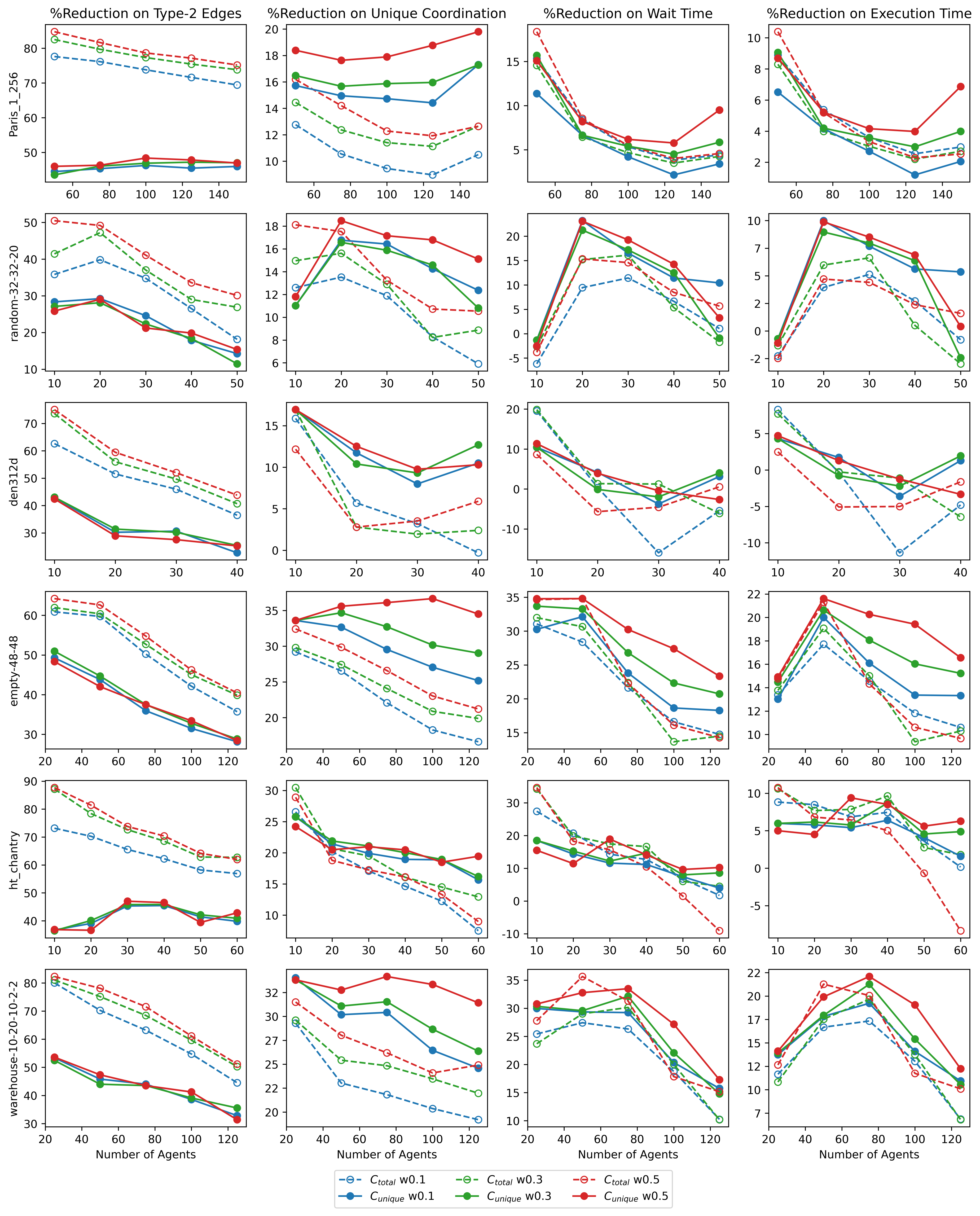}
    \caption{We compare SO-CBS with our total and unique objectives across different $w$ with ECBS-POST. ECBS-POST plans space-time paths using ECBS and then post-processes it into a TPG while our method directly produces a TPG. $w$ trades off minimizing coordination ($\uparrow w$) versus path length ($\downarrow w$). Our methods substantially decrease coordination compared to ECBS-POST (first two columns). The last two columns show SO-CBS is more robust under random delays by comparing wait time and execution time with ECBS-POST.}
    \label{fig:main-results}
    \vspace{-1em}
\end{figure*}

\begin{table*}[t]
\centering
\resizebox{0.99\textwidth}{!}{
\begin{tabular}{|c|ccc|ccc|ccc|ccc|}
\hline
\#Agents & 50 & 100 & 150 & 50 & 100 & 150 & 50 & 100 & 150 & 50 & 100 & 150 \\ \hline
Method & \multicolumn{3}{c|}{Execution Time} & \multicolumn{3}{c|}{\# Type-2 Edges} & \multicolumn{3}{c|}{Runtime (sec)} & \multicolumn{3}{c|}{Path Length} \\ \hline
ECBS-POST sub1.05   & 777 & 1287 & 1374 & 59.7 & 115 & 143 & 0.07 & 0.22 & 0.45 & 190 & 191 & 181 \\
ECBS-POST sub1.2    & 776 & 1295 & 1367 & 58.8 & 115 & 146 & 0.07 & 0.22 & 0.45 & 190 & 191 & 181 \\
SL-CBS w0.4         & 1052 & 1795 & 2151 & 41.3 & 81.9 & 108 & 0.39 & 1.49 & 3.48 & 190 & 191 & 181 \\ \hline
$C_{total}$ \ w0.1  & 708 & 1241 & 1333 & 13.4 & 30.3 & 43.9 & 1.08 & 7.16 & 48.6 & 190 & 193 & 185 \\
$C_{total}$ \ w0.3  & 712 & 1248 & 1337 & 10.4 & 26.3 & 37.6 & 1.31 & 12.1 & 92.9 & 191 & 194 & 186 \\
$C_{total}$ \ w0.5  & 696 & 1245 & 1339 & 9.10 & 24.7 & 35.7 & 1.51 & 23.1 & 112 & 192 & 197 & 190 \\ \hline
$C_{unique}$ \ w0.1 & 726 & 1253 & 1346 & 33.1 & 62.2 & 77.6 & 1.03 & 6.43 & 36.1 & 191 & 194 & 186 \\
$C_{unique}$ \ w0.3 & 706 & 1241 & 1319 & 33.6 & 61.5 & 76.0 & 1.10 & 6.46 & 45.9 & 191 & 194 & 186 \\
$C_{unique}$ \ w0.5 & 709 & 1234 & 1280 & 32.2 & 59.7 & 76.1 & 1.26 & 9.24 & 77.1 & 191 & 194 & 185 \\ \hline
\end{tabular}}
\caption{We compare average statistics of our SO-CBS methods against our baselines on Paris\_1\_256 across different numbers of agents (top row). As $w$ increases, we reduce average coordination (\# Type-2 edges) by planning slightly longer paths. This leads to an improved execution time (lower) in the face of delays, but at the expense of a longer planning time due to SO-CBS's larger search space. 
}
\label{tab:results}
\vspace{-1em}
\end{table*}


\section{Experimental Evaluation}
We experimentally show that SO-CBS produces valid TPGs that have less coordination between agents and are robust to random delays during execution. We test our methods on 6 different maps from \cite{stern2019mapfbenchmark}. For all experiments and methods, we use suboptimality $w_{so} = 1.2$ and a timeout of 2 minutes. Due to the very poor performance of the absolute order version of SO-CBS (inability to scale), we only show results of our relative order SO-CBS. We test using $C_{total}$ and $C_{unique}$ objectives with coordination cost weight $w \in \{0.1, 0.2, 0.3, 0.4, 0.5\}$. We only show results for $ w \in \{0.1, 0.3, 0.5\}$ in Figure \ref{fig:main-results} for visual clarity. The trend applies to $w \in \{0.2, 0.4\}$ as well.

Our main baseline ECBS-POST is ECBS with post-processing for the TPG. We also include SL-CBS and post-processing their space-level paths to a TPG as a second baseline. The post-processing TPG construction time is omitted as its runtime is negligible. We run ECBS with bypass conflicts and target reasoning enabled for a fair runtime comparison.
For SL-CBS, we set their cost weight $w = 0.4$, as a balanced trade-off between coordination and runtime efficiency which was the best balance they reported.

To test robustness with delays, we simulated execution with 5\% of the agents having 20\% random delays, with each delay having a duration of 100 timesteps. Agents are required to follow their TPG execution policy and we evaluate the wait time due to coordination and the total execution time summed over all agents. All results presented are averages across 10 random agent scenarios and 100 simulation seeds for sampling delays (thus each point averages 1000 runs).
We compare SO-CBS and the ECBS-POST baseline on three representative maps in Figure \ref{fig:main-results}. The SO-CBS results are shown with two different objectives and three different cost weights. Results on all five maps are shown in the appendix. The SL-CBS baseline consistently performs worse than SO-CBS and is hence omitted from Figure \ref{fig:main-results} but included in Table \ref{tab:results}.

\subsection{Reducing Coordination}
We evaluate two metrics for coordination in the TPG: the number of Type-2 edges (first column in Figure \ref{fig:main-results}), and the number of unique pairs of agents that coordinate (second column). Figure \ref{fig:main-results} highlights how Space-Order CBS significantly improves on both metrics of coordination across all maps compared to the ECBS-POST baseline. 
Optimizing $C_{total}$ reduces the total number of Type-2 edges by 30-90\% with large maps having better improvements. Optimizing $C_{unique}$ also helps in this metric although not to the same degree. For both objectives, increasing $w$ generally improves both coordination metrics. The coordination reduction also generally decreases as the number of agents increases across all maps. These observations provide evidence that SO-CBS fully utilizes free space to avoid coordinating between agents, and that in more congested maps has less flexibility to do so. The surprisingly large 80\% reductions additionally imply existing method paths have a large amount of potentially unnecessary overlap. Figure \ref{fig:motivating} visualizes this difference in two TPGs.

The main practical benefit of reducing coordination is decreasing communication between agents. Communication is a constraint that real-world multi-agent systems must handle and is non-trivial in practice with robots with limited bandwidth and range \cite{communicationReview2022}. Agents following TPGs require communication between agents to check if Type-2 edges are satisfied. Thus SO-CBS's reduction of $C_{total}$ directly decreases the communication burden for the MAPF system. We would like to highlight in conversations with two real-world warehouse companies, unreliable communication (e.g. spotty Wifi) was an issue that negatively impacted their performance. Thus, decreasing the communication load would likely help alleviate this issue.


We can trade off coordination for runtime and path length by manipulating $w$. Table \ref{tab:results} contains results for map Paris\_1\_256 whose trends are applicable to other maps. As $w$ increases from 0.1 to 0.5 with $C_{total}$, SO-CBS focuses more on reducing coordination with the improvement over ECBS-POST increasing from 70\% to 85\%, while the average path length increases only less than 5\%, implying that SO-CBS can avoid coordination by planning around. However, this ability comes at a computational cost. As shown in Section \ref{sec:lowlevelpl}, in low-level search, our method has a branching factor that is 3-10 times higher. With a larger search space, SO-CBS requires significantly longer planning time. Increasing $w$ also increases the run time as the planner needs to explore more paths to find solutions with lower coordination.

\subsection{More Robust to Delays during Execution}
Due to imperfect execution or kinematic constraints (e.g. rotating), agents can be delayed while executing the TPG plan. Individual delays can propagate to other agents that coordinate and wait for the delayed agents, incurring additional waiting. We experimentally show that SO-CBS produces TPGs that are more robust to such delays by comparing the waiting and execution time as seen in the third and fourth columns in Figure \ref{fig:main-results}. By having less coordination, SO-CBS reduces the wait time in most cases, which leads to an improvement in the overall execution time. Comparing the third and the last column, improvements in execution time preserve the trend in improvements in wait time, with values scaled down by about 30\% to 60\% depending on the contribution of wait time to the total execution time. Overall, the $C_{unique}$ objective performs better than $C_{total}$ and achieves an average of 16.2\% reduction in wait time and 9.1\% reduction in execution time at $w=0.5$. SO-CBS with $C_{unique}$ has a total execution time lower than ECBS-POST over 90\% of the time. EBCS-POST is an extremely strong state-of-the-art baseline (we ran it on a very low suboptimality and it timed out on lower ones we tried). Thus improving on average by 9.1\% is non-trivial. However as mentioned earlier, this comes at a computational and scalability cost.

Interestingly, we observe a clear positive correlation between the reduction of waiting time and $C_{unique}$. This supports our hope that with lower $C_{unique}$, each agent has fewer unique dependent agents. Subsequently, the probability that delays propagate onto dependent agents would be reduced. On the flip side, although the reduction values of $C_{total}$ are much higher, it does not align with waiting time. This supports the hypothesis that repeated coordination between the same pair of agents is over-penalized in $C_{total}$ and does not contribute to a higher probability of delay propagation than single coordination. Thus we conclude that $C_{unique}$ is a better metric that measures coordination with respect to the propagation of delays causing additional wait actions.



To further investigate how SO-CBS responds to different delay settings, we run simulations with varying ratios of delayed agents $\alpha$ and delay duration $t_d$. Figure \ref{fig:delay-duration} shows execution time reduction on empty\_48\_48 map with 50 agents with SO-CBS using $C_{unique}$ and cost weight 0.5. We expect when $t_d$ increases, there will be more delays propagated through coordination. Thus, the contribution of waiting to execution time will increase and enable more opportunities for improvement of SO-CBS over the baseline. Indeed, Figure \ref{fig:delay-duration} confirms our expectation, showing a consistent increase in reduction in execution time as $t_d$ increases. However, the contrary happens when the ratio of delayed agents $\alpha$ increases, especially when $t_d$ is large. This likely occurs as delayed agents are equally slow in expectation and are therefore less likely to propagate delays onto each other. 
But with few delayed agents, delays would almost certainly make all dependent agents that are not delayed wait. 
Nevertheless, SO-CBS is more robust across various settings of random delays. Its advantage is more significant as delay duration increases and the ratio of delayed agents decreases, with the execution time improvement as high as 30\% in some cases.

\begin{figure}[t!]
    \includegraphics[width=0.44\textwidth]{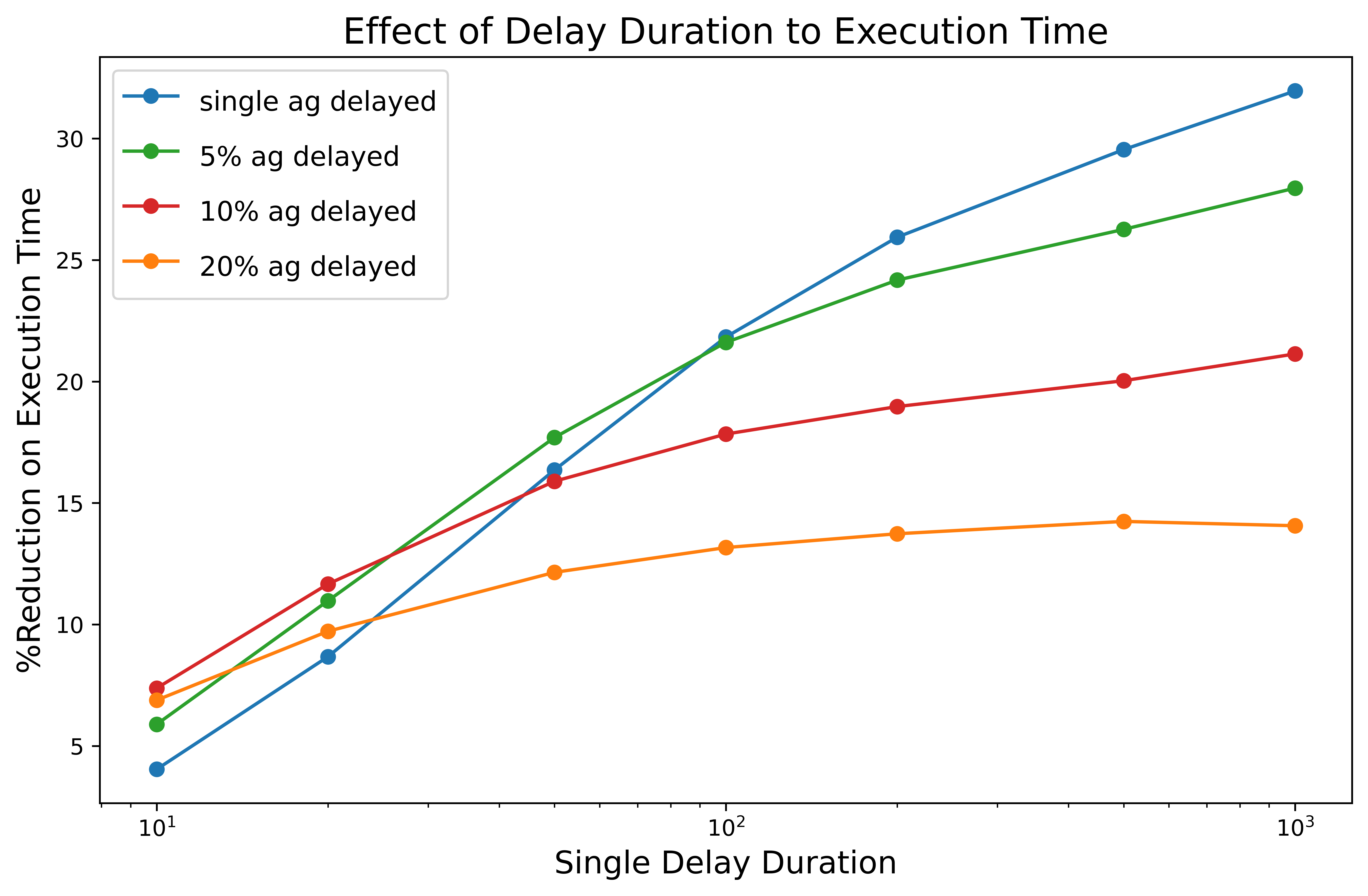}
    \vspace{-1.2em}
    \caption{empty\_48\_48\_map with 50 agents: We vary the ratio of delayed agents (colors) and delay duration (x-axis) and show the reduction of execution time (y-axis). The reduction in execution time increases as the delay duration increases.}
    \label{fig:delay-duration}
    \vspace{-1.5em}
\end{figure}

\section{Conclusion} \label{sec:conclusion}
Our work is an initial investigation into directly planning a TPG using space-visitation order planning. We first introduced our main theoretical contribution of reinterpreting TPGs as space-order paths that satisfy certain properties. We then designed unique constraints for Space-Order CBS to plan valid space-order paths, and described two variants (absolute order and relative order) with relative order SO-CBS performing convincingly better.
We experimentally demonstrate how SO-CBS can substantially reduce coordination and lead to TPGs more robust to a variety of delays.

Future work should explore different ways of computing a TPG as well as incorporate TPG post-processing methods. In particular, PIBT scales well in regular MAPF \cite{pibt} but requires modifications to satisfy space-order conflicts and coordination objectives. Additionally using more advanced TPG post-processing techniques like MCP \cite{mapfdelayprob} or BTPG \cite{su2024bidirectional} within the planning process could further decrease communication associated with coordination and increase robustness to delays. 



\bibliography{ref} 

\clearpage
\appendix

\setcounter{figure}{0}
\renewcommand{\thefigure}{A\arabic{figure}}
\setcounter{table}{0}
\renewcommand{\thetable}{A\arabic{table}}

\section{Proving Space-Order CBS is Complete}
We more thoroughly prove Theorem \ref{thm:completeness} that Space-Order CBS is weakly complete (will find a solution if one exists given enough time and memory).
As mentioned in the proof sketch, we must first approach verifying completeness in CBS's framework by proving that when resolving a conflict through multiple constraints, if a solution existed in the original CT node then a solution exists in at least one successor CT node branch (i.e. the constraints are ``complete"). Since the root node trivially allows a solution (as there are no constraints), this means that each time we resolve a conflict in a node that allows a solution, a solution will still be allowed (not constrained) in at least one successor CT node. We thus need to prove that our constraints for resolving Vertex, Edge, Target, and Deadlock Conflicts via their corresponding constraints each satisfy the completeness property. This verifies that there exists a branch in the CT that contains a solution.

The second part of the proof is to justify that the low-level and high-level search will eventually expand all the nodes in the branch containing the solution until the solution is found (i.e. the searches are complete).

\subsection{Proving Completeness of Constraints}

A common way to prove that the solution has to exist in at least one successor CT node is to prove the contrapositive; show that if a solution does not exist in any node, it is not a valid solution. If a solution does not exist in any node, this means it violates the constraints of each (and therefore all) successor node. Thus we approach the proof of completeness for each conflict and set of constraints by assuming the solution violates all constraints and showing how this cannot lead to a valid solution. We repeat the constraints with their proof of completeness.

\setcounter{definition}{8}
\begin{definition}[Vertex Constraint]
    Vertex constraints occur when two agents $j,k$ share the same $(s,r)$ location, violating Definition \ref{def:cVertex}. We add a constraint on $j$ to avoid $(s,r)$ for one child node, and a constraint on $k$ likewise (similar to resolving a regular space-time vertex conflict).
\end{definition}
\begin{proof}[Proof of completeness]
    Suppose the solution violates both constraints. Then by violating the first constraint, $j$ is at $(s,r)$. By violating the second, $k$ is at $(s,r)$ as well. This violates the vertex criteria (Definition \ref{def:cVertex}) of a valid MAPF TPG.
\end{proof}
Note: This is identical to the proof for regular CBS vertex constraints.

\begin{definition}[Target Constraint]
    This occurs when agent $k$ passes over agent $j$ resting at its goal location $(s_j^{goal},r_j)$ violating Definition \ref{def:cTarget}. We then constrain $h$ to reach and rest at $(s_j^{goal},r \leq r_j)$ and have no other agent (including $j$) with $(s_j^{goal},r > r_j)$, or constrain $j$ to reach and rest at $(s_j^{goal},r > r_j)$.
    \label{def:target-constraint}
\end{definition}
\begin{proof}[Proof of completeness]
    Violating the first constraint results in some other agent $k$ at $(s_j^{goal},r > r_j)$. Violating the second constraint means $j$ occupies $(s_j^{goal}, r \not> r_j)$. This results in a violation of Definition \ref{def:cTarget} where $j$ will conflict with $k$ resting at its goal. 
\end{proof}
Note: This constraint/proof is identical to target reasoning for space-time CBS \cite{srli2021}.

\begin{definition}[Edge Constraint]
    In the case of Definition \ref{def:cEdge}.1, two agents $j,k$ swap relative orders in consecutive locations in the same direction, e.g. $(s,r_j) \rightarrow (s',r'_j), (s,r_k > r_j) \rightarrow (s',r'_k < r'_j)$
    , violating Definition \ref{def:cEdge}. 
    Instead, we add the following child constraints: 
    (a) Constrain $j$ to avoid the set of edges $(s, r \le r_j) \rightarrow (s',r \ge r'_k)$. 
    (b) Constrain $k$ to avoid set of edges $(s,r \ge r_j) \rightarrow (s', r \le r'_k)$.
    The analogous case in Definition \ref{def:cEdge} where two agents travel in opposite directions follows without loss of generality.
\end{definition}
\begin{proof}[Proof of completeness]
    For the first edge conflict, suppose the optimal solution violates both constraints. 
    From violating the first constraint, we have $j$ going $(s, r > r_j) \rightarrow (s', r < r'_k)$. 
    From violating the second constraint, we have $k$ going $(s, r < r_j) \rightarrow (s', r > r'_k)$. 

    This results in a violation of Definition \ref{def:cEdge}. Thus the optimal solution cannot lie outside both constraints and must be in one of the successor CT nodes. The same logic holds when applied to the second conflict with agents moving in opposite directions.
\end{proof}
Note: A ``regular" edge constraint like in space-time CBS is possible and maintains completeness using the normal completeness logic. However, we found that it performed poorly and our current edge constraints boosted performance (as it is a much stronger constraint avoiding multiple edges).

\begin{definition}[Cycle Constraints]
    Given a deadlock cycle in \ref{def:cDeadlock} involving $n$ Type-2 edges: $(s^i, r^i_{i-1}) \rightarrow (s^i, r^i_i)$ where $i = 0, 1, 2, ..., n-1$ and $r^i_{i-1} < r^i_i$, we can break the cycle if any of the Type-2 edges is broken. So we add 2 constraints for each $i$:
    (a) Agent $i-1$ cannot go to $(s^i, r < r^i_i)$ if it enter location $s^{i-1}$ before $s^i$.
    (b) Agent $i$ cannot go to $(s^i, r \ge r^i_i)$.
    \label{def:deadlock-constraint}
\end{definition}
\begin{proof}[Proof of completeness]
    Assume for contradiction a solution exists that violates all $2n$ constraints. In this solution with $n$ agents $i = 0, 1, ..., n-1$, each agent $i$ must have order $r \ge r^i_i$, while agent $i-1$ have order $r < r^i_i$, i.e. agent $i$ must wait for $i-1$. Additionally each agent $i$ must first visit $s_i$, then $s_{i+1}$. We observe that these conditions reduce to the exact same condition in Definition \ref{def:cDeadlock}, implying that the same deadlock exists in this solution that violates all the constraints, which is a contradiction as a solution cannot have any deadlock cycles. Thus our $2n$ constraints and branching process for deadlock is complete.
\end{proof}

\subsection{Proving Weak Completeness of Search}
From the completeness of constraints, we know that a solution exists in the root node and will exist in at least one branch of the CT when applying constraints to generate successor CT nodes. By definition, a valid solution has a finite set of constraints and thus will be found if we expand the branch fully. We thus need to prove that the searches will eventually find the solution as opposed to something like DFSing arbitrarily bad branches.




First, we observe that since the low-level search is best first in respect to the (weighted sum) objective, the cost of every successor node must be at least as large as the cost of the parent node. This occurs as adding a constraint only reduces the search space for the low-level search.

Second, a valid solution must have a finite cost. Therefore starting from the root CT node, our high-level search will expand all nodes whose costs are less then this finite cost until it expands the solution node. Thus, Space-Order CBS will return a solution if one exists. However, if there is no solution, Space-Order CBS (like regular CBS), will keep on searching infinitely and never return no solution, and is thus only weakly complete.

We note that GCBS, ECBS, EECBS and other variants of CBS employ this same proof given different search objectives (e.g. ECBS minimizing conflicts then costs) rather than minimizing just the cost. Therefore our above Space-Order logic follows identically despite the different weighted sum objective.

\section{Details on Deadlock Conflicts}
\begin{figure}[t]
    \includegraphics[width=0.48\textwidth]{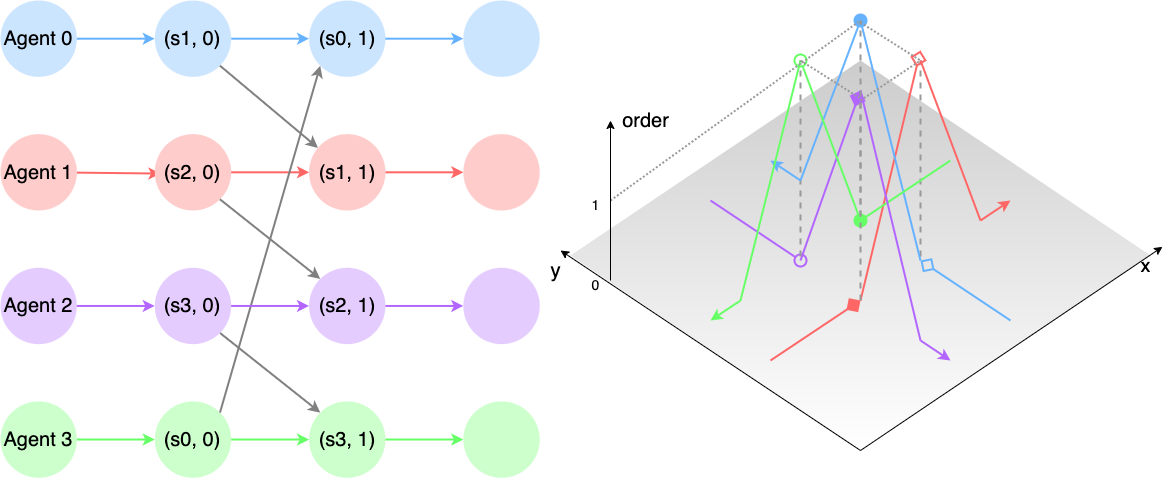}
    \caption{An example of a deadlock cycle as agents rotating in place is not allowed according to the definition of a Type-2 edge \cite{honig2016tpg}. The left subplot shows the edge perspective TPG for the scenario, with Type-1 edges denoted in color and Type-2 edges denoted in grey edges. The right subplot shows the same scenario depicted in the space-order 3D space. The blue, red, purple, and green trajectory show the space-order paths for agent 0, 1, 2, and 3 respectively. The grey vertical dotted lines represent coordination at the same location, corresponding to the grey Type-2 edges in the left subplot. There are different paired shapes at the coordinated points on the paths. The same shapes with different colors represent two agents coordinating at the same location in different orders. All agents will be stuck at the second node in their paths, as all of them have order 1 in the next node and are cyclically dependent. Concretely, green needs to wait on purple, which waits on red, which waits on blue, which in turn waits for purple, leading to deadlock.
    }
    \label{fig:deadlock-example2}
    \vspace{-1.5em}
\end{figure}

In Definition \ref{def:cDeadlock}, we define deadlock criteria with a set of $n$ agents $A$ and a set of $N$ locations $S$, which satisfy two conditions. The first condition states that each agent $a_n \in A$ needs to wait for the previous agent $a_{n-1}$ at a location $s^n \in S$. The second condition states that each agent $a_n$ must wait for the previous agent $a_{n-1}$ at $s^n$ before they can enter and leave $s^{n+1}$ to resolve the waiting for the next agent $a_{n+1}$. If both conditions are satisfied, agents in $A$ form a deadlock and will be stuck waiting for each other. 

This criteria covers most of the possible deadlocks. However, there is one additional subtle edge case of deadlock that does not fall under the above criteria. We avoid discussing this edge case in the main paper because it is technically involved. Now we discuss the edge case and present a formal definition of deadlock criteria that covers all cases. The edge case derives from the fact that TPG disallows following between agents. That is, if agent 1 waits for agent 2 at location $s$ and agent 2 \textbf{leave} $s$ at time $t$ during execution, agent 1 can only enter $s$ \textbf{after} $t$, i.e. starting from the next time step $t+1$. Now consider the example depicted in Figure \ref{fig:deadlock-example2}, following the notation in Definition \ref{def:cDeadlock}, a set of four agents $A = \{ a_n | n=0,1,2,3\}$ each occupies a location in a $2 \times 2$ grid and they try to rotate around the grid. The set of locations is $S = \{s^n |n=0,1,2,3\}$. In vertex-perspective TPG, they have paths $a_0: (s^1, 0) \rightarrow (s^0, 1)$, $ a_1 :  (s^2, 0) \rightarrow (s^1, 1)$, $a_2 : (s^3, 0) \rightarrow (s^2, 1)$ and $a_3 : (s^0, 0) \rightarrow (s^3, 1)$. Since TPG requires agents can only move to order $1$ after the agent who holds order $0$ leaves, agent $a_3$ is stuck at $(s^0, 0)$ waiting for agent $a_2$ to leave $s^3$, who is stuck at $s^3$ waiting for agent $a_1$, who waits for $a_0$ and eventually the waiting propagate to $a_3$ itself. Therefore, the 4 agents in the above situations are in a deadlock.

Let $r^n_m$ be the order of agent $a_m$ at location $s^n$. In the above situation, for $i=0,1,2,3$, $r^n_n = 1, r^n_{n-1}=0$, so $r^n_n > r^n_{n-1}$, satisfying condition 1 in Definition \ref{def:cDeadlock}. However, observe that agent $a_0$ visits $s^1$ then $s^0$, which violates condition 2 in Definition \ref{def:cDeadlock}. Note that in this situation, all agents $a_n$ visits $s^{n+1}$ then immediately visits $s^n$. Therefore, we generalize Definition \ref{def:cDeadlock} as follows:

\setcounter{definition}{11}
\begin{definition}[Complete Cycle Criteria]
A deadlock cycle of $N$ agents consists of $N$ Type-2 edges (i.e. coordination) at $N$ different locations.
Let $S = \{s^n | n = 0, 1, ... N-1 \}$ be a set of $N$ distinct locations and $A = \{a_n | n = 0, ... N-1\}$ a set of distinct agents. Index $n$ is in module $N$. Let $r^n_m$ be the visitation order of agent $a_m$ at location $s^n$. Sets $S$ and $A$ form a deadlock cycle when
\begin{enumerate}
    \item $\forall n$, agents $a_n, a_{n-1}$ visit $s^n$ with $r^n_n > r^n_{n-1}$. That is, agent $a_n$ must wait for agent $a_{n-1}$ to enter location $s^n$.
    \item $\forall n$, agent $a_n$ satisfy one of the following:
    \begin{enumerate}
        \item $a_n$ visits $s^n$ then $s^{n+1}$ 
        \item $a_n$ visits $s^{n+1}$ and then immediately visits $s^n$
    \end{enumerate}
\end{enumerate}
This forms a cycle of agents stuck waiting for each other at the Type-2 edges.
\end{definition}
Note that while in the above example, all agent falls into the edge case, it is possible for a deadlock to have part of agents in $A$ satisfy condition 2(a) while the others satisfy 2(b). Following the complete definition of deadlock criteria, we generalize Definition \ref{def:deadlock-constraint} as follows:

\begin{definition}[Complete Cycle Constraints]
    Given a deadlock cycle in \ref{def:cDeadlock} involving $n$ Type-2 edges: $(s^n, r^n_{n-1}) \rightarrow (s^n, r^n_n)$ where $n = 0, 1, 2, ..., N-1$ and $r^n_{n-1} < r^n_n$, we can break the cycle if any of the edges is broken. So we add 2 constraints for each $n$:
    (a) Agent $n-1$ cannot go to $(s^n, r < r^n_n)$ if it visit location $s^{n-1}$ then visit $s^n$ or if it visit $s^n$ then immediately visit $s^{n-1}$.
    (b) Agent $n$ cannot go to $(s^n, r \ge r^n_n)$.
    \label{def:deadlock-constraint}
\end{definition}

The proof for completeness follows identically as done previously.

\end{document}